\documentclass[11pt]{article}

\usepackage{amsmath, amssymb}

\usepackage{amsthm}

\usepackage{ifthen}

\usepackage{ifpdf}

\usepackage{subfigure}

\ifpdf
\usepackage[pdftex]{graphicx}
\usepackage[pdftex]{hyperref}
\else
\usepackage[dvips]{graphicx}
\fi

\ifthenelse{\isundefined{\NoGenerateLetterSize}}
{
\ifpdf
\setlength{\pdfpagewidth}{8.5in}
\setlength{\pdfpageheight}{11in}
\else

\fi
}
{}

\newcommand{\Comment}[1]{}

\long\def\LongVersion#1\LongVersionEnd{#1}
\long\def\ShortVersion#1\ShortVersionEnd{}

\LongVersion 
\LongVersionEnd 

\ShortVersion 
\usepackage{times}
\ShortVersionEnd 

\newtheorem{theorem}{Theorem}[section]
\newtheorem{lemma}[theorem]{Lemma}
\newtheorem{observation}[theorem]{Observation}

\theoremstyle{definition}

\newtheorem{example}[theorem]{Example}
\theoremstyle{plain}

\ShortVersion
\renewcommand{\paragraph}[1]{\par\noindent\textbf{#1}~}
\ShortVersionEnd

\newcommand{\Integers}[0]{\mathbb{Z}}
\newcommand{\Reals}[0]{\mathbb{R}}
\DeclareMathOperator{\Max}{max}
\DeclareMathOperator{\SecondMax}{max2}
\DeclareMathOperator{\Min}{min}
\newcommand{\SignalFunc}[0]{\mathcal{S}}
\newcommand{\Probability}[0]{\mathbb{P}}
\newcommand{\Expectation}[0]{\mathbb{E}}
\newcommand{\comment}[1]{}
\newcommand{\Revenue}[0]{\mathbf{Rev}}
\newcommand{\SocialWelfare}[0]{\mathbf{SW}}

\newcommand{\one}{\boldsymbol{1}}
\newcommand{\MainLp}{\textbf{LP1}}
\newcommand{\MainLpBayesian}{\textbf{LP2}}
\newcommand{\Sec}{Section}

\begin{document}
\title{Signaling Schemes for Revenue Maximization}

\author{
Yuval Emek\thanks{
ETH Zurich.
{\tt yuval.emek@tik.ee.ethz.ch}}
\and
Michal Feldman\thanks{
Hebrew University and Microsoft Research.
Partially supported by the Israel Science Foundation (grant
number 1219/09) and by the Leon Recanati Fund of the Jerusalem School of
Business Administration.
{\tt michal.feldman@huji.ac.il}}
\and
Iftah Gamzu\thanks{
Microsoft Research.
{\tt iftah.gamzu@cs.tau.ac.il}}
\and
Renato Paes Leme\thanks{
Cornell University.
Supported by a Microsoft Research  	Fellowship.
{\tt renatoppl@cs.cornell.edu}}
\and
Moshe Tennenholtz\thanks{
Microsoft Research and Technion.
{\tt moshet@microsoft.com}}
}

\date{}

\maketitle

\begin{abstract}
Signaling is an important topic in the study of asymmetric information in
economic settings.
In particular, the transparency of information available to a seller in an
auction setting is a question of major interest.
We introduce the study of signaling when conducting a second price auction of
a probabilistic good whose actual instantiation is known to the auctioneer but
not to the bidders.
This framework can be used to model impressions selling in display
advertising.
We establish several results within this framework.
First, we study the problem of computing a signaling scheme that maximizes the
auctioneer's revenue in a Bayesian setting.
We show that this problem is polynomially solvable for some interesting
special cases, but computationally hard in general.
Second, we establish a tight bound on the minimum number of signals required
to implement an optimal signaling scheme.
Finally, we show that at least half of the maximum social welfare can be
preserved within such a scheme.
\end{abstract}

\renewcommand{\thepage}{}
\clearpage
\pagenumbering{arabic}

\section{Introduction}
\label{section:Introduction}
A major concern in market design is to ensure that the markets are
\emph{thick} in the sense that there is a sufficient volume of participants to
produce the necessary level of competition for the market to work well.
Another concern is to design a practical language that is sufficiently
expressive to allow players to specify how much they value the goods in the
market.
In the market for diamonds described by Levin and
Milgrom \cite{LevinMilgrom2010} for example,
the auctioneers could elicit bids for
each individual stone.
However, the enormous effort required for the players to learn the value of
each individual stone and to submit individual bids would make the auction
impractical.
Moreover, bidding on each stone separately can lead to the
\emph{cherry-picking} phenomenon, where very few customers are interested in
any one stone.
This may lead to a situation where little revenue is generated although the
goods are valuable.
In practice, stones are categorized into \emph{deals} and then auctioned.
This method of treating different goods as identical is called
\emph{conflation}.

Milgrom \cite{Milgrom2009} and Levin and Milgrom \cite{LevinMilgrom2010}
provide a comprehensive analysis of the phenomenon of \emph{conflation} in
various markets, with particular emphasis on online advertisement.
In these markets, auctioning each good individually is usually not an option
and conflation must be used.
One particular online market that we will focus on is the multibillion-dollar
display advertisement market, where publishers (such as MSN and Yahoo) attempt
to maximize the revenue they collect from the advertisers (say, Nike or
Coca-Cola) for wisely targeting their ads at the right users.
For example, an ad referring to the surfing lifestyle on the sunny beaches of
the Pacific Ocean may be most valuable when targeted at a teenager from
California;
perhaps less so when targeted at a 10 year old from Oregon;
and even less when targeted at older folks in areas that are far from the
ocean.
However, it would be impossible for advertisers to decide how to bid on each
individual impression.
Instead, the impressions are categorized based on attributes such as the time
when the impression was made, cookies in the user's browser, certain
demographic properties, geographic location, etc.;
impressions with similar attributes are then treated as instances of the same
good.

In the context of display advertisement, the main question we deal with is:
How should those impressions be categorized in order to maximize the
publisher's revenue?
The high-level idea behind our model is to explore the natural asymmetry of
information between the publisher and the advertisers:
while advertisers may know the distribution of users visiting a particular
site, the publisher usually has much more accurate information about each
individual impression.\footnote{
In reality, the additional information about the visitor to a site is often
handled by third party demand side platform (usually refered as DSP) .
For simplicity, we abstract away this distinction.
}
Upon receiving an impression, the publisher may choose to reveal to the
advertisers certain attributes of this impression (say, age and gender), while
concealing other attributes (say, geographic location).
One may argue that concealing information from the advertisers might
generate inefficiencies in the market, but the amount of information is
typically so large that it would be impossible for the advertisers to grasp
everything anyhow.
More importantly, ensuring that revenue is generated is essential for the
proper functioning of markets.
As described by Muthukrishnan \cite{Muthukrishnan09}, in Ad Exchanges, which
are systems that bring together publishers and advertisers in a common
marketplace, ensuring good revenue is vital to keeping publishers in the market.

More generally, our goal in this paper is to cope with the undesired effects
of competition deficiency on some items in an auction.
To achieve this goal, we will exploit an inherent information asymmetry
between the auctioneer and the bidders that exists in many market settings.
We model\footnote{
For the formal exposition of our model, see Section~\ref{section:Model}.
} the auctioneer/bidders asymmetry by considering a framework termed a
\emph{probabilistic single-item auction}, in which $n$ bidders participate in
an auction for a single item, which is chosen randomly from a set
of $m$ indivisible goods according to a commonly known probability
distribution $p \in \Delta(m)$.
In contrast to the bidders, who know only the probability distribution over
the possible goods, the auctioneer knows its actual realization, and can use
this informational superiority to increase the collected revenue.

Specifically, the auctioneer may choose to reveal partial information to the
bidders by means of a \emph{signaling scheme}.
A signaling scheme is a (possibly randomized) policy that specifies some signal
$\sigma$ revealed to the bidders upon the choice (made by nature) of each good
$j \in [m]$.
This policy is known to the bidders who can therefore induce the revealed
signal $\sigma$ to update their perceived probability for the chosen good $j$
from $p(j)$ to the ``more accurate'' $p(j \mid \sigma)$.


One approach would be design an optimal auction from scratch for the problem of
maximizing revenue in a probabilistic single item auction setting. Since we are
in a setting with very correlated values, this is likely to produce contorted
and impractical auctions in the style of the auction of Cremer and McLean
\cite{cremermclean85, cremermclean88}, where full surplus extraction is
possible. Instead, we focus on the standard second price auction mechanism
which is the de-facto standard for the sale of online
advertisement \cite{Varian,Edelman,Muthukrishnan09}.
We believe this will generate an auction that is more relevant to practical
applications and that can be easily integrated with the current implementations.

In this auction, after the bidders receive the signal, they submit their bids,
and the winner and the payment are determined according to the second-price
auction; namely, the winner is the bidder with the highest bid and the payment
is the second-highest bid. The goal of the auctioneer, which is the subject of
this paper, is to design a signaling scheme that maximizes her expected revenue.

A simple but crucial observation that facilitates our analysis is that, similar
to the classical setting of second-price auctions, here too, it is a
dominant strategy for the bidders to reveal their true expected valuations,
where the expectation in this context is taken with respect to the conditional
probability $p(j \mid \sigma)$.
Therefore, the problem, termed \emph{revenue maximization by signaling},
reduces to finding a signaling scheme that maximizes the expected
second-highest bid (amounting to the expected revenue).

Two trivial signaling schemes are the one that reveals no information to the
bidders and the one that reveals the actual realization (all the information).
Interestingly, there are instances in which an appropriate signaling scheme
provides a substantial improvement over the two trivial ones.
This can be demonstrated already through a special case of a signaling scheme,
termed \emph{clustering}:
The auctioneer a-priori partitions the set of goods into disjoint clusters,
and the signal is the cluster that contains the chosen good.
Consider the case in which there are $m$ bidders and $m$ types of goods, an
item is chosen uniformly at random, and each bidder $i$ is only interested in
good $i$ with a unit valuation.
If no information is revealed then the expected revenue is $1 / m$ as the
expected valuation of each bidder is $1 / m$.
If the actual realization is revealed, no revenue is collected since for every
realization, the second-highest valuation is $0$.
However, if the goods are partitioned into clusters of size $2$, then the
expected revenue is $1 / 2$, providing an improvement of a linear factor over
the best trivial scheme.

Note that clustering schemes can be thought of as restricting the auctioneer
to deterministic policies.
The class of signaling schemes considered in this paper is more general than
clustering as we allow the auctioneer to toss coins when deciding on the
revealed signal.

\subsection*{Our Results}
We begin our analysis assuming that the valuations of the bidders are
known to the auctioneer.
In this somewhat less realistic case, the problem of revenue maximization by
signaling can be formalized as a concise linear program and, as such, solved
to optimality in polynomial time.
A natural question is to what extent the increase in revenue comes at the
expense of social welfare.
Notably, we prove that a signaling scheme that obtains the optimal revenue can
preserve at least half of the optimal social welfare.
In addition, it is shown that if the auctioneer is restricted to invoking a
signaling scheme by means of clustering, at least half of the optimal revenue
can be achieved, and this is tight.

Up until now we have assumed that the valuations of the bidders are known to
the auctioneer.
However, in practice, the auctioneer rarely knows the bidders' valuations.
This motivates the main technical contribution of this paper, namely, the
study of a \emph{Bayesian} setting, in which the auctioneer holds 
probabilistic knowledge on the bidders' valuations.
We show that in this case the revenue maximization by signaling problem
becomes NP-hard.
Still, in several cases of interest the problem remains tractable even in
the Bayesian setting.
Finally, we show that $m$ signals are always sufficient to extract the optimal
revenue.
It is an interesting open problem how to find good approximation algorithms for
the cases where the revenue maximization is NP-hard or to prove hardness of
approximation in those cases. 

Notice that our model captures the Bayesian knowledge on behalf of the
auctioneer by assuming a probability distribution over finitely many valuation
matrices.
This representation can capture complicated dependencies between the different
valuations, however it may exhibit plenty of redundancy when the valuations
are assumed to be independent.
As such, it will be interesting to study our framework under more concise
representations -- for example, where each entry of the matrix is sampled
independently from some distribution.

\subsection*{Related Work}
There is a rich theory on markets with information asymmetry.
In such markets, agents on one side have more (or better) information than
those on the other side.
The foundation of this theory dates back to the work of Akerlof, Spence, and
Stiglitz on the analysis of markets with asymmetric information, which
earned them the 2001 Nobel Prize.
In particular, Akerlof \cite{Akerlof70} introduced the first formal analysis of
markets in which sellers have more information than buyers regarding the
quality of products.
Spence \cite{Spence73,Spence02} demonstrated that in certain settings,
well-informed agents can improve their outcome by signaling their private
information to poorly informed agents.

There is also a vast literature on the nature and effects of information
revelation in auctions.
One of the most fundamental results in auction theory, namely the ``Linkage
Principle'' of Milgrom and Weber \cite{MilgromWeber}, states that the expected
revenue of an auctioneer is enhanced when bidders are provided with more
information.
While this work advocates transparency in various markets, later work observed
that such  transparency may not be optimal in general (see,
e.g.,~\cite{PerryReny99,Krishna02,FoucaultL03,FeinbergT05}).
More recent work \cite{Milgrom2009,LevinMilgrom2010} advocated the need for
careful grouping of goods as an important market design principle.
Our work may be viewed as a study of information revelation through
an optimization lens, since we seek to maximize the expected revenue of an
auctioneer by designing an effective information revelation scheme.

Notice that Myerson's classic result on revenue maximization \cite{myerson}
does not apply to our model due to the asymmetry of information.
Also, Myerson's mechanism works only for single parameter settings.
Our Bayesian models are multi-parameter and are typically highly correlated.
Revenue maximization results for correlated valuations \cite{pp, dfk} also do
not
apply here due to asymmetric information.
One could try to reveal all information and then apply one of those
mechanisms, but we would get a rather contorted auction, with no guarantees
against our auction.
In fact, it is easy to construct examples where this generates arbitrary less
revenue than our signaling scheme.
Our auction, on the other hand, is very practical and close to what is
actually implemented in online advertising markets.

Closer to our work is that of Ghosh et al \cite{ghosh}, which studies
revenue-maximizing clustering schemes under a second-price auction in a
setting with full information and additive valuations.
While this setting is different from our framework of signaling in
a probabilistic item auction, the mathematical formulation of the optimization
problem in their setting is a special case of our optimization
problem, i.e., the case where the valuation matrix is known to the auctioneer
and the signaling scheme is restricted to take the form of a clustering
scheme.
Our focus, though, is on the more realistic Bayesian case which
is not not treated in \cite{ghosh}.
In addition, our framework relies on signaling that can be viewed as a
``fractional'' clustering, which is more powerful.
Indeed, while Ghosh et al. show that it is strongly NP-hard to compute the
optimal clustering scheme, an optimal signaling scheme can be computed in
polynomial time.
Results of the similar flavour of the ones in \cite{ghosh} were re-derived
independently in a previous version of the current paper \cite{Gamzu_aaw}.

Independently of our work, Miltersen and Sheffet \cite{Miltersen_Sheffet} also
analyze the problem of
obtaining optimal signaling schemes for revenue maximization using Linear
Programming, obtaining a result similar to our Theorem \ref{theorem:OptimalLp}.

\section{The Model}
\label{section:Model}
In this section, we introduce the auctioning model on top of which our
signaling schemes are defined.
Our focus in this paper is on a Bayesian setting, treating the uncertainty of
the auctioneer regarding the bidders' valuations in a probabilistic manner.
For clarity of the exposition, we shall first consider the (less realistic)
\emph{known-valuations} setting, where no such uncertainty is assumed.

\subsection*{Known-Valuations Probabilistic Single-Item Auctions}
A \emph{known-valuations probabilistic single-item auction (KPSA)}
$\mathcal{A}$ is formally depicted by the four-tuple
$$
\mathcal{A} = \left\langle n, m, p, V \right\rangle \ ,
$$
where $n \in \Integers_{> 0}$ stands for the number of \emph{bidders},
$m \in \Integers_{> 0}$ stands for the number of distinct indivisible
\emph{goods},
$p \in \Delta(m)$ is a probability distribution over the goods, and
$V \in \Reals_{\geq 0}^{n \times m}$ is a non-negative real matrix capturing
the \emph{valuation} $V(i, j)$ of bidder $i$ for good $j$.
A single good $j \in [m]$ is chosen (by nature) according to the distribution
$p$ which is a common knowledge.

The auction is conducted according to the \emph{second-price} rule:
Each player $i$ places her bid $b_i$ and the chosen good $j$ is sold to the
bidder that placed the highest bid $\Max_{i \in [n]} \{ b_i \}$ (ties are broken
arbitrarily) for the price of the second highest bid $\SecondMax_{i \in [n]}
\{ b_i \}$.

\subsection*{Signaling Schemes}
Although the bidders know the distribution $p$, they do not know its actual
realization which is observed only by the \emph{auctioneer}.
In an attempt to increase her expected \emph{revenue}, the auctioneer may
partially reveal the realization $j \in [m]$ of $p$ to the bidders.
This partial revelation is carried out by means of \emph{signaling}:
given that the chosen good is $j$ (recall that this choice is made by nature),
the auctioneer sends the bidders some signal $\sigma$;
the bidders then hold a ``more accurate picture'' of the chosen good that
corresponds to the probability distribution $p$ conditioned on $\sigma$.
The policy that dictates the signal that the auctioneer reveals to the bidders
for each good $j \in [m]$ is referred to as a \emph{signaling scheme}.
It is important to point out that this policy is decided by the auctioneer and
reported to the bidders prior to nature's random choice of item $j$.

More formally, a signaling scheme is given by a set of $s \in \Integers_{>
0}$ \emph{signals} and a \emph{signaling function} $\SignalFunc : [s] \times
[m] \rightarrow [0, 1]$ that satisfies
\begin{equation} \label{equation:ValidSignalingScheme}
\sum_{\sigma \in [s]} \SignalFunc(\sigma, j) = 1
\quad
\forall j \in [m] \ .
\end{equation}
Given that nature chose good $j \in [m]$, the auctioneer reveals signal
$\sigma \in [s]$ to the bidders with probability $\SignalFunc(\sigma, j)$.
It will be convenient to use the notation $\SignalFunc$ to address the
signaling scheme as well as its inherent signaling function.

Once again, it is assumed that $\SignalFunc$ (and $s$) are decided by the
auctioneer and reported to the bidders prior to the random choice of $j$;
it is the actual signal $\sigma$ (determined according to $\SignalFunc$) that
is revealed to the bidders after the choice of $j$.
(This can be thought of as a commitment of the auctioneer to stick to the
signaling scheme that it previously reported.)
Upon receiving signal $\sigma$, the bidders, knowing $p$ and $\SignalFunc$,
update their belief from $\Probability(\text{chosen good is $j$}) = p(j)$ to
\begin{align*}
\Probability(\text{chosen good is $j$} \mid \text{signal is $\sigma$})
~ = ~ &
\frac{\Probability(\text{signal is $\sigma$} \mid \text{chosen good is $j$})
\cdot \Probability(\text{chosen good is $j$})}{\Probability(\text{signal is
$\sigma$})} \\
= ~ &
\frac{\SignalFunc(\sigma, j) \cdot p(j)}{\sum_{j' \in [m]} \SignalFunc(\sigma,
j') \cdot p(j')} \ .
\end{align*}
For succinctness, we will subsequently denote the events ``$\text{chosen good is
$j$}$'' and ``$\text{signal is $\sigma$}$'' by $j$ and $\sigma$,
respectively (our intention will be clear from the context).

Before we proceed, let us consider the restricted variant of a
\emph{deterministic auctioneer} which is not allowed to use randomness when
determining which signal to reveal.
This is equivalent to imposing an additional ``integrality'' requirement on
the signaling scheme:
$\SignalFunc(\sigma, j) \in \{0, 1\}$ for every $\sigma \in [s]$ and $j \in
[m]$.
In other words, each signal $\sigma \in [s]$ now corresponds to a
\emph{cluster} $C_{\sigma} \subseteq [m]$ so that the clusters are pairwise
disjoint and $\bigcup_{\sigma \in [s]} C_{\sigma} = [m]$.
Following this view, the general case (under which the auctioneer may use
randomness when determining the signal) can be interpreted as a fractional
clustering of the goods, where $\SignalFunc(\sigma, j)$ is the fraction of
good $j$ in cluster $C_{\sigma}$.

It is well known that in the classical setting of second-price single-item
auctions, it is a dominant strategy for the bidders to be truthful, i.e., to
bid their true valuations~\cite{vickery1961}.
It turns out that this remains valid in probabilistic single-item auctions
under signaling as well, as the following observation demonstrates
(proof deferred to the appendix).

\begin{observation} \label{observation:SecondPrice}
For every $i \in [n]$ and $\sigma \in [s]$, bidding $b_i(\sigma) =
\Expectation[V(i, j) \mid \sigma]$ in response to the signal $\sigma$
is a dominant strategy for bidder $i$.
\end{observation}
\newcommand{\ProofObservationSecondPrice}{
Consider the known-valuations \emph{ex-ante} game defined by setting the
strategy space of each bidder $i \in [n]$ to be the collection of all possible
functions $b_i : [s] \rightarrow \Reals_{\geq 0}$ and the utility of bidder
$i$ from strategy profile $b = (b_1, \dots, b_n)$ to be the expected utility
of bidder $i$ in the KPSA $\mathcal{A}$ assuming that each bidder adheres to
$b$.
Fix some (arbitrary) strategies $b_{i'} : [s] \rightarrow \Reals_{\geq 0}$ for
all bidders $i' \neq i$ and consider the strategy $b_i$ of bidder $i$ that
bids
$$
b_{i}(\sigma) =
\Expectation \left[ V(i, j) \mid \sigma \right] =
\sum_{j \in [m]} \Probability(j \mid \sigma) \cdot V(i, j) =
\sum_{j \in [m]} \frac{\SignalFunc(\sigma, j) \cdot p(j)}{\sum_{j' \in [m]}
\SignalFunc(\sigma, j') \cdot p(j')} \cdot V(i, j)
$$
in response to each signal $\sigma \in [s]$.

Fix some signal $\sigma$.
From bidder $i$'s point of view, the expected valuation of the chosen good is
$b_{i}(\sigma)$, whereas each other bidder $i' \neq i$, bids $b_{i'}(sigma)$.
If bidder $i$ does not win the chosen good, which happens only if $\Max_{i'
\neq i} \{ b_{i'}(\sigma) \} \geq b_{i}(\sigma)$, then her utility in the
ex-ante game is $0$.
This can be changed only if bidder $i$ increases her bid so that it exceeds
$\Max_{i' \neq i} \{ b_{i'}(\sigma) \}$, but this imposes a negative utility on
$i$.
So, assume that $\Max_{i' \neq i} \{ b_{i'}(\sigma) \} \leq b_{i}(\sigma)$ and
bidder $i$ does win the chosen good.
By the definition of the second-price rule, the utility of $i$ must be
non-negative.
Clearly, bidder $i$ has no incentive to increase her bid.
Decreasing her bid does not change her utility as long as it still exceeds
$\Max_{i' \neq i} \{ b_{i'}(\sigma) \}$;
decreasing her bid further resets her utility to zero.
The assertion follows.
} 

\subsection*{Optimization Problems}
Consider some KPSA $\mathcal{A} = \langle n, m, p, V \rangle$ and signaling
scheme $\SignalFunc$.
In light of Observation~\ref{observation:SecondPrice}, we subsequently assume
that the bidders are indeed truthful, that is, bidder $i$ bids
$\Expectation[V(i, j) \mid \sigma] = \sum_{j \in [m]} \Probability(j \mid \sigma)
\cdot V(i, j)$ in response to the signal $\sigma$.
Therefore, the (expected) \emph{revenue} of the auctioneer, denoted
$\Revenue_{\mathcal{A}}(\SignalFunc)$, is given by
$$
\Revenue_{\mathcal{A}}(\SignalFunc)
~ = ~
\sum_{\sigma \in [s]} \Probability(\sigma) \cdot \SecondMax_{i \in [n]}
\left\{ \sum_{j \in [m]} \Probability(j \mid \sigma) \cdot V(i, j) \right\} \ .
$$
This raises the following optimization problem, referred to as the
\emph{revenue maximization by signaling (RMS) problem}:
given a KPSA $\mathcal{A}$, construct the signaling scheme $\SignalFunc$
that maximizes $\Revenue_{\mathcal{A}}(\SignalFunc)$.
One may also be interested in the (expected) \emph{social welfare} resulting
from signaling scheme $\SignalFunc$, defined as
$$
\SocialWelfare_{\mathcal{A}}(\SignalFunc)
~ = ~
\sum_{\sigma \in [s]} \Probability(\sigma) \cdot \Max_{i \in [n]}
\left\{ \sum_{j \in [m]} \Probability(j \mid \sigma) \cdot V(i, j) \right\} \ .
$$
When $\mathcal{A}$ is clear form the context, we may omit it from the
subscript and write simply $\Revenue(\SignalFunc)$ and
$\SocialWelfare(\SignalFunc)$.

Notice that the revenue of the auctioneer can be rewritten as
\begin{align}
\Revenue(\SignalFunc)
~ = ~ &
\sum_{\sigma \in [s]} \Probability(\sigma) \cdot \SecondMax_{i \in [n]}
\left\{ \sum_{j \in [m]} \frac{\Probability(\sigma \mid j) \cdot
\Probability(j)}{\Probability(\sigma)} \cdot V(i, j) \right\}
\nonumber \\
= ~ &
\sum_{\sigma \in [s]} \SecondMax_{i \in [n]} \left\{ \sum_{j \in [m]}
\SignalFunc(\sigma, j) \cdot p(j) \cdot V(i, j) \right\} \nonumber \\
= ~ &
\sum_{\sigma \in [s]} \SecondMax_{i \in [n]} \left\{ \sum_{j \in [m]}
\SignalFunc(\sigma, j) \cdot \Psi(i, j) \right\} \
, \label{equation:AlternativeRevenueDefinition}
\end{align}
where $\Psi(i, j) = p(j) \cdot V(i, j)$ is referred to as the
\emph{normalized valuation} of bidder $i$ for item $j$.
Following the same line of arguments, we can also rewrite the social welfare as
$$
\SocialWelfare(\SignalFunc)
~ = ~
\sum_{\sigma \in [s]} \Max_{i \in [n]} \left\{ \sum_{j \in [m]}
\SignalFunc(\sigma, j) \cdot \Psi(i, j) \right\} \ .
$$

Under the deterministic auctioneer requirement, the RMS problem turns into
the following clustering problem:
Given the normalized valuation matrix $\Psi \in \Reals_{\geq 0}^{n \times m}$,
devise a pairwise disjoint partition of $[m]$ into clusters $\{ C_{\sigma}
\}_{\sigma \in [s]}$ that maximizes
$$
\Revenue \left( \left\{ C_{\sigma} \right\}_{\sigma \in [s]} \right) =
\sum_{\sigma \in [s]} \SecondMax_{i \in [n]} \left\{ \sum_{j \in C_{\sigma}}
\Psi(i, j) \right\} \ .
$$

\subsection*{A Bayesian Setting}
Recall that up until now, we assumed that the valuations of the bidders are
known to the auctioneer.\footnote{
In some sense, we also assumed that the valuations of each bidder are known to
the other bidders.
However, Observation~\ref{observation:SecondPrice} implies that
this does not matter:
a bidder is better off bidding its true (expected) valuation regardless of the
strategies of the other bidders.
}
However, in many practical scenarios the auctioneer does not know the exact
valuation of each bidder.
To tackle this obstacle, we assume a \emph{Bayesian} setting, treating the
state of knowledge that the auctioneer holds on the bidders' valuations in a
probabilistic manner.
This is captured in our model by considering $k \in \Integers_{> 0}$ distinct
valuation matrices $V_1, \dots, V_k \in \Reals_{\geq 0}^{n \times m}$ and a
probability distribution $q \in \Delta(k)$ associating each valuation matrix
$V_{\ell}$ with the probability $q(\ell)$ that it occurs.
A \emph{probabilistic single-item auction (PSA)} is then depicted by the
$6$-tuple
$$
\mathcal{A} = \left\langle n, m, k, p, q, \{ V_{\ell} \}_{\ell \in [k]}
\right\rangle \ ,
$$
where $n \in \Integers_{> 0}$, $m \in \Integers_{> 0}$, and $p \in \Delta(m)$
have the same role as in the known-valuations case; and
$k \in \Integers_{> 0}$, $q \in \Delta(k)$, and $\left\{ V_{\ell} \in \Reals_{\geq
0}^{n \times m} \right\}_{\ell \in [k]}$ capture the aforementioned Bayesian
angle.

The expected revenue of the auctioneer from
the signaling scheme $\SignalFunc$ is now defined to be
\begin{align*}
\Revenue_{\mathcal{A}}(\SignalFunc)
~ = ~ &
\sum_{\ell \in [k]} q(\ell)
\sum_{\sigma \in [s]}
\Probability(\sigma) \cdot \SecondMax_{i \in [n]}
\left\{ \sum_{j} \Probability(j \mid \sigma) \cdot V_{\ell}(i, j) \right\} \\
= ~ &
\sum_{\ell \in [k]} q(\ell)
\sum_{\sigma \in [s]}
\SecondMax_{i \in [n]} \left\{ \sum_{j}
\SignalFunc(\sigma, j) \cdot \Psi_{\ell}(i, j) \right\} \ ,
\end{align*}
where $\Psi_{\ell}(i, j) = p(j) \cdot V_{\ell}(i, j)$ and the last equation
follows from the same line of arguments that was used to establish
(\ref{equation:AlternativeRevenueDefinition}).

\section{Optimal Signaling Schemes --- The Known-Valuations Case}
\label{section:OptimalSchemesKnownValuations}
Let us start our technical treatment of signaling schemes with (the simpler)
known-valuations setting, considering a KPSA $\mathcal{A} = \langle n, m, p, V
\rangle$.
We show that an optimal signaling scheme for $\mathcal{A}$ can be obtained by
solving an LP with $O (n^2 m)$ variables and $O (n^2 + m)$ constraints
(excluding the non-negativity constraints).

Given an $s$-signal signaling scheme $\SignalFunc$ for $\mathcal{A}$ and a
signal $\sigma \in [s]$, let $h_{1}^{\SignalFunc}(\sigma)$ and
$h_{2}^{\SignalFunc}(\sigma)$ denote the bidders $i$ that realize
$\Max_{i \in [n]} \left\{ \sum_{j \in [m]} \SignalFunc(\sigma, j) \cdot
\Psi(i, j) \right\}$ and
$\SecondMax_{i \in [n]} \left\{ \sum_{j \in [m]} \SignalFunc(\sigma, j) \cdot
\Psi(i, j) \right\}$, respectively.
(When the signaling scheme $\SignalFunc$ is clear from the context, we may
omit the superscripts.)
Our concise LP relies on the following observation.

\begin{observation} \label{observation:FewSignals}
There exists an optimal $s$-signal signaling scheme $\SignalFunc$ for
$\mathcal{A}$ such that given $\sigma, \sigma' \in [s]$, if
$h_{1}^{\SignalFunc}(\sigma) = h_{1}^{\SignalFunc}(\sigma')$ and
$h_{2}^{\SignalFunc}(\sigma) = h_{2}^{\SignalFunc}(\sigma')$, then $\sigma =
\sigma'$.
\end{observation}
\begin{proof}
Consider an optimal $s$-signal signaling scheme $\SignalFunc$ that minimizes
$s$.
We argue that $\SignalFunc$ must satisfy the assertion.
To that end, assume by contradiction that there are two distinct signals
$\sigma, \sigma' \in [s]$ such that $h_{1}^{\SignalFunc}(\sigma) =
h_{1}^{\SignalFunc}(\sigma') = i_1$ and $h_{2}^{\SignalFunc}(\sigma) =
h_{2}^{\SignalFunc}(\sigma') = i_2$.
Let $\SignalFunc^*$ be the $(s - 1)$-signal signaling scheme obtained from
$\SignalFunc$ by replacing both signals $\sigma$ and $\sigma'$ by a new signal
$\sigma^*$ defined by setting $\SignalFunc(\sigma^*, j) = \SignalFunc(\sigma,
j) + \SignalFunc(\sigma', j)$ for every $j \in [m]$.

It is easy to verify that $\SignalFunc^*$ is valid in terms of
(\ref{equation:ValidSignalingScheme}).
Moreover, since $h_{1}^{\SignalFunc^*}(\sigma^*) = i_1$ and
$h_{2}^{\SignalFunc^*}(\sigma^*) = i_2$, we can use
(\ref{equation:AlternativeRevenueDefinition}) to conclude that the combined
contribution of $\sigma$ and $\sigma'$ to $\Revenue(\SignalFunc)$ is
$$
\sum_{j \in [m]} \left( \SignalFunc(\sigma, j) + \SignalFunc(\sigma', j)
\right) \cdot \Psi(i_2, j)
~ = ~
\sum_{j \in [m]} \SignalFunc^*(\sigma^*, j) \cdot \Psi(i_2, j)
$$
which is precisely the contribution of $\sigma^*$ to
$\Revenue(\SignalFunc^*)$.
Thus, $\Revenue(\SignalFunc) = \Revenue(\SignalFunc^*)$, in contradiction to
the minimality of $s$.
\end{proof}

A direct corollary of Observation~\ref{observation:FewSignals} is that it
suffices to consider signaling schemes with $s = n (n - 1)$ signals --- each
signal $\sigma$ is uniquely identified by $h_{1}(\sigma)$ and
$h_{2}(\sigma)$.
This turns out to be asymptotically tight as there are examples showing that
$\Omega (n^2)$ signals are required to implement an optimal signaling scheme
(see \Sec{}~\ref{section:BoundNumberSignals}).
Based on Observation~\ref{observation:FewSignals} and on the formulation of
revenue in (\ref{equation:AlternativeRevenueDefinition}), we can construct an
optimal signaling scheme $\SignalFunc$ by solving the following linear
program, denoted \MainLp{}:
\begin{align*}
\max ~ & \sum_{i_1, i_2 \in [n], i_1 \neq i_2} R(\sigma_{i_1, i_2}) ~
\text{s.t.} \\
& R(\sigma_{i_1, i_2}) ~ \leq ~ \sum_{j \in [m]} \SignalFunc(\sigma_{i_1,
i_2}, j) \cdot \Psi(i_1, j) \quad \forall i_1, i_2 \in [n], i_1 \neq i_2 \\
& R(\sigma_{i_1, i_2}) ~ = ~ \sum_{j \in [m]} \SignalFunc(\sigma_{i_1,
i_2}, j) \cdot \Psi(i_2, j) \quad \forall i_1, i_2 \in [n], i_1 \neq i_2 \\
& \sum_{i_1, i_2 \in [n], i_1 \neq i_2} \SignalFunc(\sigma_{i_1, i_2}, j) ~ =
~ 1 \quad \forall j \in [m] \\
& \SignalFunc(\sigma_{i_1, i_2}, j) ~ \geq ~ 0 \quad \forall i_1, i_2 \in [n],
i_1 \neq i_2, \forall j \in [m] \ .
\end{align*}

Since \MainLp{} consists of $O (n^{2} m)$ variables and $O (n^{2} +
m)$ constraints (excluding the non-negativity constraints), it can be solved
in polynomial time.

\begin{theorem} \label{theorem:OptimalLp}
Under the known-valuations setting, the RMS problem can be solved in
polynomial time.
\end{theorem}

\subsection*{Signaling versus Clustering}
A clustering scheme is a special case of a signaling scheme, where the
auctioneer cannot use randomness (see Section~\ref{section:Model}).
This restricted case has been studied in \cite{ghosh}, and is equivalent to
imposing the requirement that $\SignalFunc(\sigma,j) \in \{0,1\}$ for every
$\sigma$ and $j$ in our framework.

\begin{theorem}
The optimal revenue that can be extracted by a signaling scheme is at most
twice the optimal revenue that can be extracted by a clustering scheme, and
this is tight.
\end{theorem}

\begin{proof}
The algorithm in \cite{ghosh} produces a clustering scheme that extracts
revenue that is greater or equal to half of
\begin{equation}
\label{eq:upper_bound_ghosh}
\min_{i'} \sum_j \max_{i \neq i'} \Psi(i,j).
\end{equation}
Therefore, in order to establish the upper bound, it is sufficient to show
that the revenue extracted by any signaling scheme is bounded by Equation
(\ref{eq:upper_bound_ghosh}).
For every $i' \in [n]$, one can express the revenue of a signaling scheme
$\SignalFunc$ as
\begin{align*}
\Revenue(\SignalFunc)
~ = ~ &
\sum_\sigma \SecondMax_i \sum_j \SignalFunc(\sigma,j)
\Psi(i,j) \\
\leq ~ & \sum_\sigma \max_{i \neq i'} \sum_j \SignalFunc(\sigma,j)
\Psi(i,j) \\
\leq ~ &
\sum_\sigma  \sum_j \SignalFunc(\sigma,j) \max_{i \neq i'}
\Psi(i,j) \\
= ~ & \sum_j\max_{i \neq i'} \Psi(i,j) \ ,
\end{align*}
where the last equality holds since $\sum_\sigma \SignalFunc(\sigma,j) = 1$.
The upper bound follows.

To establish the lower bound, consider Example~\ref{gap_example}.
While the optimal signaling scheme extracts revenue $\frac{n}{(n+1)}$, it is
not difficult to verify that the optimal clustering scheme partitions items
$1, \ldots, m$ into pairs and leaves item 0 as a singleton.
This clustering scheme extracts revenue $\frac{n}{2(n+1)}$, which is half of
the revenue extracted by the optimal signaling scheme.
\end{proof}

We remark that while we used Equation (\ref{eq:upper_bound_ghosh}) as our
benchmark, a better benchmark would be to compare the clustering revenue to
the solution of the LP in Theorem~\ref{theorem:OptimalLp}. However,
Example~\ref{gap_example} demonstrates that $\frac{1}{2}$ is tight with
respect to this benchmark as well.

\subsection*{Social Welfare versus Revenue}
Increasing the revenue by signaling usually comes at the expense of degrading
the social welfare.
We show, however, that it is easy to calculate the best revenue one can get
without degrading the social welfare by much.
For every $j \in [m]$, let $\mu(j)$ denote the bidder $i$ that maximizes the
normalized valuation $\Psi(i, j)$ (which means that $i$ also maximizes $V(i,
j)$).
Then, the optimal social welfare is given by $W^* = \sum_{j \in [m]}
\Psi(\mu(j), j)$.
By augmenting \MainLp{} with the constraint
$$
\sum_{i_1, i_2 \in [n], i_1 \neq i_2} \, \sum_{j \in [m]}
\SignalFunc(\sigma_{i_1, i_2}, j) \cdot \Psi(i_1, j)
~ \geq ~
\beta W^* \ ,
$$
we guarantee the highest possible revenue conditioned on preserving at least a
$\beta$-fraction of the social welfare.
Theorem~\ref{theorem:SocialWelfare} (whose proof is deferred to the appendix)
shows that taking $\beta \leq 1 / 2$ does not affect \MainLp{}.
Note that this theorem can be viewed as a signaling analogue of Theorem~2 in
\cite{ghosh} and its proof essentially follows similar arguments.

\begin{theorem} \label{theorem:SocialWelfare}
There exists a revenue-optimal signaling scheme $\SignalFunc$ with
$\SocialWelfare(\SignalFunc) \geq W^* / 2$.
\end{theorem}
\newcommand{\ProofTheoremSocialWelfare}{
Consider an optimal signaling scheme and let $\sigma$ be a signal and $j$ an
item such that $ \SignalFunc(\sigma, j) > 0$.
If $\mu(j) \notin \{h_{1}^{\SignalFunc}(\sigma),
h_{2}^{\SignalFunc}(\sigma)\}$, then we construct the new signaling scheme
$\widehat{\SignalFunc}$ obtained from $\SignalFunc$ by replacing signal
$\sigma$ with the two new signals $\sigma', \sigma''$ such that
$\widehat{\SignalFunc}(\sigma', j) = \SignalFunc(\sigma,
j)$, $\widehat{\SignalFunc}(\sigma', j') = 0$ for $j' \neq j$; and
$\widehat{\SignalFunc}(\sigma'', j) = 0$, $\SignalFunc(\sigma'', j') =
\SignalFunc(\sigma, j')$ for $j' \neq j$.

We argue that $\Revenue(\widehat{\SignalFunc}) \geq \Revenue(\SignalFunc)$
which, by the optimality of $\SignalFunc$, implies that
$\Revenue(\widehat{\SignalFunc}) = \Revenue(\SignalFunc)$.
To that end, note that the contribution of $\sigma$ to $\Revenue(\SignalFunc)$
is
\begin{align*}
\Revenue(\sigma, \SignalFunc)
~ = ~ &
\SignalFunc(\sigma, j) \cdot \Psi(h_{2}^{\SignalFunc}(\sigma), j) +
\sum_{j' \neq j} \SignalFunc(\sigma, j') \cdot
\Psi(h_{2}^{\SignalFunc}(\sigma), j') \\
\leq ~ &
\SignalFunc(\sigma, j) \cdot \Psi(h_{1}^{\SignalFunc}(\sigma), j) +
\sum_{j' \neq j} \SignalFunc(\sigma, j') \cdot
\Psi(h_{1}^{\SignalFunc}(\sigma), j') \ ,
\end{align*}
whereas the contributions of $\sigma'$ and $\sigma''$ to
$\Revenue(\widehat{\SignalFunc})$ are
$$
\Revenue(\sigma', \widehat{\SignalFunc})
~ = ~
\SignalFunc(\sigma, j) \cdot \SecondMax_{i \in [n]} \left\{ \Psi(i, j) \right\}
$$
and
$$
\Revenue(\sigma'', \widehat{\SignalFunc})
~ = ~
\SecondMax_{i \in [n]} \left\{ \sum_{j' \neq j} \SignalFunc(\sigma, j') \cdot
\Psi(i, j') \right\} \ ,
$$
respectively.
The argument follows since
$$
\Revenue(\sigma', \widehat{\SignalFunc})
~ \geq ~
\Max \left\{ \SignalFunc(\sigma, j) \cdot \Psi(h_{2}^{\SignalFunc}(\sigma), j),
\SignalFunc(\sigma, j) \cdot \Psi(h_{1}^{\SignalFunc}(\sigma), j) \right\}
$$
and
$$
\Revenue(\sigma'', \widehat{\SignalFunc})
~ \geq ~
\Min \left\{ \sum_{j' \neq j} \SignalFunc(\sigma, j') \cdot
\Psi(h_{2}^{\SignalFunc}(\sigma), j'), \sum_{j' \neq j} \SignalFunc(\sigma, j') \cdot
\Psi(h_{1}^{\SignalFunc}(\sigma), j') \right\} \ .
$$

It follows that there exists a revenue-optimal signaling scheme $\SignalFunc$
such that $\SignalFunc(\sigma, j) > 0$ only if $\mu(j) \in \{ h_{1}(\sigma),
h_{2}(\sigma) \}$.
Therefore, the social welfare of $\SignalFunc$ satisfies
\begin{align*}
\SocialWelfare(\SignalFunc)
~ = ~ &
\sum_{\sigma \in [s]} \sum_{j \in [m]} \SignalFunc(\sigma, j) \cdot
\Psi(h_{1}(\sigma), j) \\
\geq ~ &
\sum_{\sigma \in [s]} \sum_{j \in [m]} \SignalFunc(\sigma, j) \cdot
\frac{\Psi(h_{1}(\sigma), j) + \Psi(h_{2}(\sigma), j)}{2} \\
\geq ~ &
\frac{1}{2} \sum_{\sigma \in [s]} \sum_{j \in [m]} \SignalFunc(\sigma, j)
\cdot \Psi(\mu(j), j)
~ = ~
\frac{1}{2} W^* ~ .
\end{align*}
The assertion follows.
} 

\section{Optimal Signaling Schemes --- The Bayesian Case}
We now turn to discuss the more interesting Bayesian setting, considering a
PSA $\mathcal{A} = \left\langle n, m, k, p, q, \{ V_{\ell} \}_{\ell \in [k]}
\right\rangle$, where $q$ is a probability distribution over the valuation
matrices $V_1, \dots, V_k \in \Reals_{\geq 0}^{n \times m}$.
Our goal in this section is twofold:
(1) proving that the RMS problem under the Bayesian setting is NP-hard; and
(2) presenting poly-time algorithms when $k$ or $m$ are fixed.
Note that the RMS problem remains NP-hard if $n$ is fixed as long as both $k$
and $m$ are free parameters.

\subsection*{Tractable Special Cases}
Let us start with developing an efficient algorithm for the RMS problem
assuming that $k = O (1)$ (without any restriction on $n$ or $m$).
Consider some $s$-signal signaling scheme $\SignalFunc$ for $\mathcal{A}$.
Given a Bayesian outcome $\ell \in [k]$ and a signal $\sigma \in [s]$, let
$h_{1}^{\SignalFunc}(\ell, \sigma)$ and $h_{2}^{\SignalFunc}(\ell, \sigma)$
denote the bidders $i$ that realize $\Max_{i \in [n]} \left\{ \sum_{j \in [m]}
\SignalFunc(\sigma, j) \cdot \Psi_{\ell}(i, j) \right\}$ and $\SecondMax_{i
\in [n]} \left\{ \sum_{j \in [m]} \SignalFunc(\sigma, j) \cdot \Psi_{\ell}(i,
j) \right\}$, respectively.
(When the signaling scheme $\SignalFunc$ is clear from the context, we may
omit the superscripts.)
Using this notation, we can now state the following observation which is
established by repeating the line of arguments that led to
Observation~\ref{observation:FewSignals}.

\begin{observation} \label{observation:FewSignalsBayesian}
There exists an optimal $s$-signal signaling scheme $\SignalFunc$ for
$\mathcal{A}$ such that given $\sigma, \sigma' \in [s]$, if
$h_{1}^{\SignalFunc}(\ell, \sigma) = h_{1}^{\SignalFunc}(\ell, \sigma')$ and
$h_{2}^{\SignalFunc}(\ell, \sigma) = h_{2}^{\SignalFunc}(\ell, \sigma')$ for
every $\ell \in [k]$, then $\sigma = \sigma'$.
\end{observation}

Observation~\ref{observation:FewSignalsBayesian} implies that it is sufficient
to consider $O (n^{2 k})$ signals $\sigma$, each uniquely identified by
$h_{1}(1, \sigma), h_{2}(1, \sigma), \dots, h_{1}(k, \sigma), h_{2}(k,
\sigma)$.
In order to formulate it as a concise linear program, we fix
$$
\Lambda =
\left\{ \left\langle (i_{1}^{1}, i_{2}^{1}), \dots, (i_{1}^{k}, i_{2}^{k})
\right\rangle \mid
i_{h}^{\ell} \in [n] \, \forall h \in \{1, 2\}, \ell \in
[k] \, \land \, i_{1}^{\ell} \neq i_{2}^{\ell} \, \forall \ell \in [k]
\right\} ~ .
$$
An optimal signaling scheme $\SignalFunc$ can now be constructed by solving
the following linear program, denoted \MainLpBayesian{}:
\begin{align*}
\max ~ & \sum_{\lambda \in \Lambda} R(\sigma_{\lambda}) ~
\text{s.t.} \\
& R(\sigma_{\lambda})
~ \leq ~
\sum_{\ell \in [k]} q(\ell) \sum_{j \in [m]}
\SignalFunc(\sigma_{\lambda}, j) \cdot
\Psi_{\ell}(\lambda(\ell, 1), j) \quad \forall \lambda \in \Lambda \\
& R(\sigma_{\lambda})
~ = ~
\sum_{\ell \in [k]} q(\ell) \sum_{j \in [m]}
\SignalFunc(\sigma_{\lambda}, j) \cdot
\Psi_{\ell}(\lambda(\ell, 2), j) \quad \forall \lambda \in \Lambda \\
& \sum_{\lambda \in \Lambda} \SignalFunc(\sigma_{\lambda}, j)
~ = ~
1 \quad \forall j \in [m] \\
& \SignalFunc(\sigma_{\lambda}, j) ~ \geq ~ 0 \quad \forall \lambda \in
\Lambda, \forall j \in [m] \ .
\end{align*}

Since \MainLpBayesian{} consists of $O (n^{2 k} m)$ variables and $O (n^{2 k} +
m)$ constraints (excluding the non-negativity constraints), it can be solved
in polynomial time as long as $k$ is constant.

\begin{theorem} \label{theorem:OptimalLpBayesianFixedK}
If $k$ is fixed, then the RMS problem can be solved in polynomial time.
\end{theorem}

Next, we show how to compute an optimal $s$-signal signaling scheme
$\SignalFunc$ when $m = O (1)$ (without any restriction on $k$ and $n$).
The main ingredient for this will be the following lemma (refer to
\cite{Stanley} for a proof).

\begin{lemma} \label{lemma:whitney}
The number of distinct regions with non-empty interior\footnote{
A region $R \in \Reals^m$ is said to have a non-empty interior if it contains
an $m$-dimensional open set.
} defined by $t \geq m$ hyperplanes in $\Reals^m$ is bounded from above by the
Whitney number $W(m, t) = \sum_{i=0}^m {t \choose i} = O(t^m)$.
\end{lemma}

Given some $\lambda \in \Lambda$, we define $X_{\lambda}$ to be the region
that contains every vector $x \in \Reals_{\geq 0}^{m} - \{0\}$ such that
\begin{equation} \label{equation:LinearConstraints}
\sum_{j \in [m]} x(j) \cdot \Psi_{\ell} \left( \lambda(\ell, 1), j \right)
~ \geq ~
\sum_{j \in [m]} x(j) \cdot \Psi_{\ell} \left( \lambda(\ell, 2), j \right)
~ \geq ~
\sum_{j \in [m]} x(j) \cdot \Psi_{\ell} \left( i, j \right)
\end{equation}
for every $\ell \in [k]$ and $i \notin \{ \lambda(\ell, 1), \lambda(\ell, 2)
\}$.
The key observation here is that if two signals $\sigma, \sigma' \in [s]$ are
such that their corresponding vectors $\SignalFunc(\sigma, \cdot),
\SignalFunc(\sigma, \cdot) \in \Reals_{\geq 0}^{m}$ fall into the same region
$X_{\lambda}$, $\lambda \in \Lambda$, then we can merge them without
decreasing the revenue.
Therefore, if we can come up with a poly-size subset $\Lambda'
\subseteq \Lambda$ so that the regions in $\{ X_{\lambda} \mid \lambda \in
\Lambda' \}$ cover the entire $\Reals_{\geq 0}^{m} - \{0\}$, then we can
construct an optimal signaling scheme by picking one signal (the right one)
for each region $X_{\lambda}$ such that $\lambda \in \Lambda'$.

So, how can we come up with such a subset $\Lambda' \subseteq \Lambda$?
It turns out that although there are many regions $X_{\lambda}$, only a
polynomially small subset of them have a non-empty interior.
Indeed, the total number of linear constraints
(\ref{equation:LinearConstraints}) involved in the definition of the regions
$X_{\lambda}$, $\lambda \in \Lambda$, is $n^2 k$ (each linear constraint is of
the form $\sum_{j \in [m]} x(j) \cdot \Psi_{\ell}(i, j) \geq \sum_{j \in [m]}
x(j) \cdot \Psi_{\ell}(i', j)$ for some $i, i' \in [n]$ and $\ell \in [k]$).
Since those linear constraints correspond to hyperplanes in $\Reals^m$,
Lemma~\ref{lemma:whitney} guarantees that there are $O ((n^2 k)^m)$ regions
$X_{\lambda}$ with a non-empty interior.

Once the subset $\Lambda' \subseteq \Lambda$ of regions with non-empty
interior has been identified, providing the linear constraints of
(\ref{equation:LinearConstraints}) for each such region, we can rewrite
\MainLpBayesian{}, dedicating a single signal $\sigma$ to each region
$X_{\lambda}$ such that $\lambda \in \Lambda'$ (the vector
$\SignalFunc(\sigma, \cdot)$ takes the role of the vector $x$ in
(\ref{equation:LinearConstraints})).
The resulting linear program consists of $O (|\Lambda'| \cdot m)$ variables and
$O (|\Lambda'| \cdot k n)$ constraints, thus it can be solved in polynomial
time.

It remains to show that we can efficiently enumerate the collection of regions
with non-empty interiors.
This is carried out by recursion on $k$:
When $k = 1$, we can simply iterate through all the regions and check if their
interior is non-empty.
For the recursive step, observe that if the region $X_{\lambda}$ has an empty
interior, then clearly, so does the region $X_{\lambda \circ (i_{1}^{k + 1},
i_{2}^{k + 1})}$ for every $i_{1}^{k + 1}, i_{2}^{k + 1} \in [n]$.
Therefore, we can keep iterating only through those regions that had a
non-empty interior in the previous recursive level.
By Lemma~\ref{lemma:whitney}, the whole process requires checking $n (n - 1)
\cdot O ((n^2 (k - 1))^m)$ regions.

\begin{theorem} \label{theorem:OptimalLpBayesianFixedM}
If $m$ is fixed, then the RMS problem can be solved in polynomial time.
\end{theorem}

\subsection*{Hardness of the General Case}
Finally, we establish the NP-completeness of (the decision version of) the RMS
problem in the Bayesian setting for the case $n = 3$.
The inclusion of this problem in NP follows from
Theorem~\ref{theorem:MSignals} that ensures that it suffices to
consider signaling schemes with at most $m$ signals (which also implies that
the number of bits required to represent the solutions of the LP is
polynomial).

The remainder of the section is dedicated to proving that the RMS problem in
the Bayesian setting is NP-hard.
This is done by a reduction from MAX-CUT (problem ND16 in
\cite{GareyJohnson79}):
Given a graph $G = (V, E)$ and two vertices $x, y \in V$, the MAX-CUT problem
asks for the maximum integer $k$ such that there exists a vertex subset $U
\subseteq V$, $|\{x, y\} \cap U| = 1$, with at least $k$ edges crossing
between $U$ and $V - U$.
Given such an instance of MAX-CUT, assuming that $|V| = n$ and $|E| = m$,
we construct a PSA $\mathcal{A}$ with $3$ bidders, $n$ items (associated with
 the vertices in $V$), and $2 n + m -
3$ Bayesian outcomes.
It will be convenient to associate the Bayesian outcomes as follows:
each $2 \leq \ell \leq n - 1$ is associated with some vertex $u \in V - \{x,
y\}$;
each $n \leq \ell \leq 2 n - 3$ is also associated with some vertex $u \in V -
\{x, y\}$;
each $2 n - 2 \leq \ell \leq 2 n + m - 3$ is associated with some edge $(u, v)
\in E$.
Figure~\ref{fig:reduction} depicts the values of $\Phi_{\ell}(i, j) = q(\ell)
\cdot \Psi_{\ell}(i, j)$ for every $i \in \{1, 2, 3\}$, $j \in [n]$, and $\ell
\in [2 n + m - 3]$, where $K_1 \gg K_2 \gg 1$ are integers that will be
determined in the course of the proof.
This specifies everything we need for the reduction.

\begin{figure}
\centering
\subfigure[$\ell = 1$]{
\begin{tabular}{ c || c | c }
    & $s$ & $t$ \\
  \hline
  $1$ & $K_1$ & \\
  $2$ & & $K_1$  \\
  $3$ & $K_1$ & $K_1$\\
\end{tabular}
} \quad
\subfigure[$2 \leq \ell \leq n - 1$]{
\begin{tabular}{ c || c | c }
    & $s$ & $u$ \\
  \hline
  $1$ & $K_2$ & \\
  $2$ & & $K_2$  \\
  $3$ &  & \\
\end{tabular}
} \quad
\subfigure[$n \leq \ell \leq 2 n - 3$]{
\begin{tabular}{ c || c | c }
    & $t$ & $u$ \\
  \hline
  $1$ & $K_2$ & \\
  $2$ & & $K_2$  \\
  $3$ &  & \\
\end{tabular}
} \quad
\subfigure[$2 n - 2 \leq \ell \leq 2 n + m - 3$]{
\begin{tabular}{ c || p{1cm} | p{1cm} }
    & $u$ & $v$ \\
  \hline
  $1$ &  $1$ & \\
  $2$ & &  $1$  \\
  $3$ & $1$ & $1$\\
\end{tabular}
}
\caption{ \label{fig:reduction}
Representation of the mapping from MAX-CUT to Bayesian signaling.
The tables represent $\Phi_{\ell}(i, v)$ for $v \in V$ and $i \in \{1, 2,
3\}$.
Values not specified in the tables are zero.}
\end{figure}

Suppose that the solution to the MAX-CUT instance is $C^*$ and that this
is realized by the vertex subset $X \subseteq V$, where $x \in X$ and $y \in Y
= V - X$.
We design a signaling scheme $\SignalFunc$ with two signals $\sigma_x,
\sigma_y$ such that
$
\SignalFunc(\sigma_x, u) =
\left\{
\begin{array}{ll}
1 & \text{if } u \in X \\
0 & \text{if } u \in Y
\end{array}
\right.$
and
$\SignalFunc(\sigma_y, u) =
\left\{
\begin{array}{ll}
0 & \text{if } u \in X \\
1 & \text{if } u \in Y
\end{array}
\right.$.
It can be checked that the revenue generated by this signaling scheme is
$$
\sum_{\ell \in [2 n + m - 3]} \sum_{\sigma \in \{ \sigma_x, \sigma_y \}}
\SecondMax_{i \in \{1, 2, 3\}} \left\{ \sum_{u \in V} \Phi_{\ell}(i, u) \cdot
\SignalFunc(\sigma, u) \right\}
~ = ~
2 K_1 + (n - 2) K_2 + m + C^* ~ .
$$
The reduction is completed by showing that this is an upper bound on the
revenue generated by any signaling scheme.

Consider some $s$-signal signaling scheme $\SignalFunc$.
The revenue of $\SignalFunc$ is
\begin{align}
\Revenue(\SignalFunc)
~ = ~ &
\sum_{\ell \in [2 n + m - 3]} \sum_{\sigma \in [s]}
\SecondMax_{i \in \{1, 2, 3\}} \left\{ \sum_{u \in V} \Phi_{\ell}(i, u) \cdot
\SignalFunc(\sigma, u) \right\} \nonumber \\
= ~ &
K_1 \sum_{\sigma \in [s]} \Max\{ \SignalFunc(\sigma, x), \SignalFunc(\sigma,
y) \} \nonumber \\
& + K_2 \sum_{u \in V - \{x, y\}} \sum_{\sigma \in [s]} \left[
\Min\{ \SignalFunc(\sigma, x), \SignalFunc(\sigma, u) \} +
\Min\{ \SignalFunc(\sigma, y), \SignalFunc(\sigma, u) \}
\right] \nonumber \\
& + \sum_{(u, v) \in E} \sum_{\sigma \in [s]} \Max\{ \SignalFunc(\sigma, u),
\SignalFunc(\sigma, v) \} \ . \label{equation:RevenueReduction}
\end{align}
We argue that if $\SignalFunc$ is optimal, then for each signal $\sigma$,
either $\SignalFunc(\sigma, x) = 0$ or $\SignalFunc(\sigma, y) = 0$.
Indeed, if there is a signal $\sigma$ with $\SignalFunc(\sigma, x)$ and
$\SignalFunc(\sigma, y)$ simultaneously positive, then we can split this
signal into two signals $\sigma', \sigma''$ such that $\SignalFunc(\sigma', x)
= \SignalFunc(\sigma, x)$ and $\SignalFunc(\sigma', u) = 0$ for $u \neq x$;
$\SignalFunc(\sigma'', x) = 0$ and $\SignalFunc(\sigma'', u) =
\SignalFunc(\sigma, u)$ for $u \neq x$.
Taking $K_1$ to be sufficiently large ensures that the revenue increases
following this transformation.

Now, consider a signal $\sigma \in [s]$ with $\SignalFunc(\sigma, x) > 0$.
If there exists some $u \neq x$ such that $\SignalFunc(\sigma, u) >
\SignalFunc(\sigma, x)$, then there must exists a signal $\sigma' \in [s]$
such that $\SignalFunc(\sigma', u) < \SignalFunc(\sigma', x)$ since
$\sum_{\sigma \in [s]} \SignalFunc(\sigma, x) = \sum_{\sigma \in [s]}
\SignalFunc(\sigma, u) = 1$.
Therefore, we can increase the value of $\SignalFunc(\sigma', u)$ and decrease
the value of $\SignalFunc(\sigma, u)$ by a small value $\delta$, obtaining a
valid signaling scheme with larger revenue (this is due to the fact that $K_2$
is large compared to $1$).
Similarly, we can claim that in an optimal signaling scheme, a signal $\sigma
\in [s]$ with $\SignalFunc(\sigma, y) > 0$ has $\SignalFunc(\sigma, u) \leq
\SignalFunc(\sigma, y)$ for all $u \in V$.\
The same argument also implies that for each signal $\sigma$ that has positive
probability, either $\SignalFunc(\sigma, x) > 0$ or $\SignalFunc(\sigma, y) >
0$.

Consider a signal $\sigma$ and let $x = u_0, u_1, \dots, u_k \in V$ be the
items with positive $\SignalFunc(\sigma, u_i)$.
Assume without loss of generality that
$\SignalFunc(\sigma, u_0) \geq \SignalFunc(\sigma, u_1) \geq \cdots \geq
\SignalFunc(\sigma, u_k) > \SignalFunc(\sigma, u_{k + 1}) = 0$.
Now, split signal $\sigma$ in $k + 1$ signals $\sigma^0, \sigma^1, \dots,
\sigma^k$ such that $\SignalFunc(\sigma^i, s) = \SignalFunc(\sigma^i, u_1) =
\cdots = \SignalFunc(\sigma^i, u_i) = \SignalFunc(\sigma, u_i) -
\SignalFunc(\sigma, u_{i+1})$.
By substituting the split signals into (\ref{equation:RevenueReduction}), we
conclude that the revenue is kept unchanged.

Therefore, PSA instances produced by our reduction always admit an
optimal $s$-signal signaling scheme $\SignalFunc$ such that for every signal
$\sigma \in [s]$, there exist a vertex subset $U_{\sigma} \subseteq V$ and a
real $0 < p_{\sigma} \leq 1$ satisfying: \\
(1) $|U_{\sigma} \cap \{x, y\}| = 1$; \\
(2) $\SignalFunc(\sigma, u) = p_{\sigma}$ for every $u \in U_{\sigma}$; \\
(3) $\SignalFunc(\sigma, u) = 0$ for every $u \notin U_{\sigma}$; \\
(4) $\sum_{\sigma : v \in U_{\sigma}} p_{\sigma} = 1$ for every $v \in V$; and \\
(5) $\sum_{\sigma \in [s]} p_{\sigma} = \sum_{\sigma : x \in U_{\sigma}}
p_{\sigma} + \sum_{\sigma : y \in U_{\sigma}} p_{\sigma} = 2$.
This is employed in order to prove that $\Revenue(\SignalFunc) \leq 2 K_1 + (n
- 2) K_2 + m + C^*$.
From (\ref{equation:RevenueReduction}), we see that
$\Revenue(\SignalFunc) =
2 K_1 + (n - 2) K_2 +
\sum_{(u, v) \in E} \sum_{\sigma \in [s]}
\Max\{\SignalFunc(\sigma, u), \SignalFunc(\sigma, v) \}$,
so it remains to show that
\begin{equation} \label{equation:RemainingTerm}
\sum_{(u, v) \in E} \sum_{\sigma \in [s]} \Max\{\SignalFunc(\sigma, u),
\SignalFunc(\sigma, v) \} \leq m + C^* \ .
\end{equation}

To that end, note that every edge $(u, v) \in E$ and signal $\sigma \in [s]$
contribute $p_{\sigma}$ to the the left-hand side in
(\ref{equation:RemainingTerm}) if $|\{u, v\} \cap U_{\sigma}| \geq 1$;
and $0$ otherwise.
Therefore, the left-hand side in (\ref{equation:RemainingTerm}) is equal to
$$
\sum_{(u, v) \in E} \left[ 1 + (1 / 2) \sum_{\sigma : | \{u, v\} \cap
U_{\sigma} | = 1} p_{\sigma} \right]
~ = ~
m + \sum_{\sigma \in [s]} \frac{p_{\sigma}}{2} |\partial(U_{\sigma})|
~ \leq ~
m + C^* \ ,
$$
where $\partial(U)$ is the set of edges with exactly one endpoint in $U$, and
the last inequality follows since
$\sum_{\sigma \in [s]} \frac{p_{\sigma}}{2} = 1$, hence $\sum_{\sigma \in [s]}
\frac{p_{\sigma}}{2} |\partial(U_{\sigma})|$ can be viewed as the average size
of cuts corresponding to the vertex subsets $U_{\sigma}$, $\sigma \in [s]$.
This establishes the following theorem.

\begin{theorem}
The decision version of the RMS problem in the Bayesian setting is
NP-complete.
The problem is hard already for $n = 3$.
\end{theorem}

\section{Bounding the Number of Signals}
\label{section:BoundNumberSignals}
In a market with $n$ bidders, the algorithm described in
\Sec{}~\ref{section:OptimalSchemesKnownValuations} generates a signaling
scheme with $O(n^2)$ signals.
In fact, some instances might require that many signals in order to produce
the optimal signaling.

\begin{example} \label{example:ManySignals}
Consider a KPSA with $n^2 - n$ items appearing with uniform
probability.
Each item $I_{i, i'}$ is labeled with an ordered pair $(i, i')$ of bidders,
where $V(i, I_{i, i'}) = 1$, $V(i', I_{i, i'}) = \frac{1}{2}$, and $V(i'',
I_{i, i'}) = 0$ for every $i'' \in [n] - \{i, i'\}$.
The optimal signaling generates revenue of $\frac{3}{4}$, by emitting a signal
$\sigma_{i, i'}$ when either $I_{i, i'}$ or $I_{i', i}$ is chosen.
Notice that $s = {n \choose 2}$ signals are required to implement this
signaling scheme.
It is not hard to see that with fewer signals, it is impossible to achieve
this revenue.
\end{example}

The KPSA described in Example~\ref{example:ManySignals} has $m = \Omega
(n^2)$ items.
Can we construct a similar bad example with much fewer items?
More generally, can we bound the minimum number of signals required to
implement an optimal signaling scheme as a function depending only on $m$?

\begin{theorem}\label{theorem:MSignals}
Every PSA has an optimal signaling scheme with $m$ signals.
\end{theorem}
\begin{proof}
For brevity, we establish the assertion assuming the known-valuations
setting;
the proof of the Bayesian setting follows from the same line of arguments.
Consider some KPSA $\mathcal{A} = \langle n, m, p, V \rangle$ and let
$\SignalFunc$ be an $s$-signal signaling scheme for $\mathcal{A}$ that
minimizes $s$.
Assume by contradiction that $s > m$.
Associate with each signal $\sigma \in [s]$ a vector $\vec{q}_{\sigma}
\in \Reals_{\geq 0}^{m}$, where $q_{\sigma}(j) = \SignalFunc(\sigma, j)$ for
every $j \in [m]$.
We can rewrite the contribution of $\sigma$ to the revenue of $\SignalFunc$ as
$\Revenue(\sigma, \SignalFunc) = \SecondMax_{i \in [n]} \{ \sum_{j \in [m]}
q_{\sigma}(j) \cdot \Psi(i, j) \}$, so $\Revenue(\cdot, \SignalFunc)$ is a
homogeneous, but not necessarily linear, operator.

Since $s > m$, the vector collection $\{ \vec{q}_{\sigma} \mid \sigma \in [s]
\}$ must exhibit linear dependencies.
Therefore, there must exist some reals $x_1, \dots, x_r, x_{r + 1}, \dots,
x_{t} > 0$ and signals $\sigma_1, \dots \sigma_r, \sigma_{r + 1}, \dots,
\sigma_{t}$
such that
$$
x_1 \vec{q}_{\sigma_1} + \cdots + x_r \vec{q}_{\sigma_r}
~ = ~
x_{r + 1} \vec{q}_{\sigma_{r + 1}} + \cdots + x_{t} \vec{q}_{\sigma_{t}} \ .
$$
Consider the $s$-signal signaling scheme $\SignalFunc'$ defined by
the modified signals $\sigma'$ obtained from $\SignalFunc$ by setting
$$
\vec{q}_{\sigma'_z}
~ = ~
\left\{
\begin{array}{ll}
(1 + \epsilon x_z) \vec{q}_{\sigma_z} & \text{if } 1 \leq z \leq r \\
(1 - \epsilon x_z) \vec{q}_{\sigma_z} & \text{if } r + 1 \leq z \leq t \\
\vec{q}_{\sigma_z} & \text{otherwise} \ .
\end{array}
\right.
$$
For any $\epsilon \in \left[ -\frac{1}{\Max_{1 \leq z \leq r} \{x_z\}},
\frac{1}{\Max_{r + 1 \leq z \leq t} \{x_z\}} \right]$, the resulting signaling
scheme $\SignalFunc'$ is valid as $\sum_{\sigma' \in [s]} \vec{q}_{\sigma'} =
\sum_{\sigma \in [s]} \vec{q}_{\sigma} = \one$ and $\vec{q}_{\sigma'} \geq 0$.

Now, notice that the revenue of the new signaling scheme satisfies
$$
\Revenue(\SignalFunc')
~ = ~
\Revenue(\SignalFunc) +
\epsilon \left[ \sum_{z = 1}^{r} x_z \Revenue(\sigma_z, \SignalFunc) - \sum_{z
= r + 1}^{t} x_z \Revenue(\sigma_z, \SignalFunc) \right] \ .
$$
Since $\SignalFunc$ is optimal, it must be the case that $\sum_{z = 1}^{r} x_z
\Revenue(\sigma_z, \SignalFunc) - \sum_{z = r + 1}^{t} x_z \Revenue(\sigma_z,
\SignalFunc) = 0$ (recall that $\epsilon$ can be taken to be positive
or negative).
Thus, we can take $\epsilon$ to be either $-\frac{1}{\Max_{1 \leq z \leq r}
\{x_z\}}$ or $\frac{1}{\Max_{r + 1 \leq z \leq t} \{x_z\}}$ and get an optimal
signaling scheme with less than $s$ signals, in contradiction to the choice of
$\SignalFunc$.
\end{proof}

If one can achieve revenue $R$ with $s$ signals, then it is trivial to see that
with $s' \leq s$ signals, one can achieve revenue $\left\lfloor \frac{s'}{s}
\right\rfloor R$.
This turns out to be the best possible in some cases.

\begin{example}\label{gap_example}
Consider $n + 1$ items $\{0, 1, \dots, n\}$, each chosen with probability
$\frac{1}{n + 1}$ and $n + 1$ bidders with valuations $V(i, i) = 1$ for $i =
1, \dots, n$, $V(0, 0) = n$, and $V(i, j) = 0$ otherwise.
The optimal signaling scheme $\SignalFunc$ uses $n$ signals $\sigma_1, \dots,
\sigma_n$ such that $\SignalFunc(\sigma_i, i) = 1$ and $\SignalFunc(\sigma_i,
0) = \frac{1}{n}$ for $i = 1, \dots, n$ and zero otherwise.
The revenue of this scheme is $\frac{n}{n + 1}$.
Now, for any $s < n$ and an $s$-signal signaling scheme $\SignalFunc'$, we have
\begin{align*}
\Revenue(\SignalFunc')
~ = ~ &
\sum_{\sigma' \in [s]} \SecondMax_{i \in [n]} \left\{ \sum_{j \in [m]}
\SignalFunc'(\sigma', j) \cdot \Psi(i, j) \right\} \\
\leq ~ &
\sum_{\sigma' \in [s]} \Max_{i \in [n] - \{0\}} \left\{ \sum_{j \in [m]}
\SignalFunc'(\sigma', j) \cdot \Psi(i, j) \right\} \\
\leq ~ &
\sum_{\sigma' \in [s]} \frac{1}{n + 1}
~ = ~
\frac{s}{n + 1} \ .
\end{align*}
\end{example}

\bibliographystyle{abbrv}
\bibliography{bibfile}

\LongVersion 
\clearpage
\pagenumbering{roman}
\appendix

\begin{figure}[t]
\begin{center}
\textbf{\large{APPENDIX}}
\end{center}
\end{figure}

\begin{proof}[Proof of Observation~\ref{observation:SecondPrice}]
\ProofObservationSecondPrice{}
\end{proof}

\begin{proof}[Proof of Theorem~\ref{theorem:SocialWelfare}]
\ProofTheoremSocialWelfare{}
\end{proof}
\LongVersionEnd 

\end{document}